\documentclass[a4paper]{article}

\usepackage[margin=1in]{geometry}

\usepackage[utf8]{inputenc}
\usepackage{amsmath,hyperref,amssymb, mathtools,amsthm}
\usepackage{subcaption}
\usepackage{paralist}
\usepackage[textsize=tiny,textwidth=3.5cm,color=green!50]{todonotes}
\usepackage{thmtools} 
\usepackage{thm-restate}
\usepackage{academicons}
\usepackage{xcolor}

\newtheorem{theorem}{Theorem}

\newtheorem{corollary}[theorem]{Corollary}

\def\RR{{\mathbb R}}

\def\SSS{{\mathcal S}}

\usepackage{xspace}

\newcommand{\kk}{\texorpdfstring{$k$}{k}}

\newcommand{\ua}{{\ensuremath{\uparrow}}}
\newcommand{\da}{{\ensuremath{\downarrow}}}

\makeatletter
\def\Gt{%
   \@ifnextchar[%
     {\mycommand@i}
     {\mycommand@i[t]}%
}
\def\mycommand@i[#1]{%
   \@ifnextchar[%
     {\mycommand@ii{#1}}
     {\mycommand@ii{#1}[c, \circ]}%
}
\def\mycommand@ii#1[#2]{%
    \ensuremath{
  G_{#1}^{#2}\xspace
  }
}

\def\Cp{%
   \@ifnextchar[%
     {\mycommandd@i}
     {\mycommandd@i[t]}%
}
\def\mycommandd@i[#1]{%
   \@ifnextchar[%
     {\mycommandd@ii{#1}}
     {\mycommandd@ii{#1}[c, \circ]}%
}
\def\mycommandd@ii#1[#2]{%
    \ensuremath{
    C_{#1}^{#2,\text{rem}}\xspace
  }
}

\def\Cclient{%
   \@ifnextchar[%
     {\mycommanddd@i}
     {\mycommanddd@i[t]}%
}
\def\mycommanddd@i[#1]{%
   \@ifnextchar[%
     {\mycommanddd@ii{#1}}
     {\mycommanddd@ii{#1}[c, \circ]}%
}
\def\mycommanddd@ii#1[#2]{%
    \ensuremath{
    C_{#1}^{#2}\xspace
  }
}

\makeatother

\newboolean{ViewTextBool}
\setboolean{ViewTextBool}{false}

\newcommand{\ViewText}[1]
{\ifthenelse{\boolean{ViewTextBool}}{#1}{}}

 \usepackage{lineno}

\begin{document}



\title{Clustering Graphs of Bounded Treewidth\\ to Minimize the Sum of Radius-Dependent Costs}

\author{
Lukas Drexler\thanks{Department of Computer Science, Heinrich-Heine-University, Germany. Email: \texttt{lukas.drexler@hhu.de}. Funded by the Deutsche Forschungsgemeinschaft (DFG, German Research Foundation) – Project Number 456558332.}
\and 
Jan Höckendorff\thanks{Department of Mathematics and Computer Science, University of Cologne, Germany. Email: \texttt{hoeckendorff@cs.uni-koeln.de}. Funded by the Deutsche Forschungsgemeinschaft (DFG, German Research Foundation) – Project Number 459420781.}
\and
Joshua Könen\thanks{Institute of Computer Science, University of Bonn, and Lamarr Institute for Machine Learning and Artificial Intelligence, Germany. Email: \texttt{s6jjkoen@uni-bonn.de}}
\and
Kevin Schewior\thanks{Department of Mathematics and Computer Science, University of Southern Denmark, Denmark. Email: \texttt{kevs@sdu.dk}. Partially supported by the Independent Research Fund Denmark, Natural Sciences, grant DFF-0135-00018B.}}

\date{}




\maketitle

\begin{abstract}
We consider the following natural problem that generalizes min-sum-radii clustering: Given is $k\in\mathbb{N}$ as well as some metric space $(V,d)$ where $V=F\cup C$ for facilities $F$ and clients $C$. The goal is to find a clustering given by $k$ facility-radius pairs $(f_1,r_1),\dots,(f_k,r_k)\in F\times\mathbb{R}_{\geq 0}$ such that $C\subseteq B(f_1,r_1)\cup\dots\cup B(f_k,r_k)$ and $\sum_{i=1,\dots,k} g(r_i)$ is minimized for some increasing function $g:\mathbb{R}_{\geq 0}\rightarrow\mathbb{R}_{\geq 0}$. Here, $B(x,r)$ is the radius-$r$ ball centered at $x$. For the case that $(V,d)$ is the shortest-path metric of some edge-weighted graph of bounded treewidth, we present a dynamic program that is tailored to this class of problems and achieves a polynomial running time, establishing that the problem is in \textsf{XP} with parameter treewidth.
\end{abstract}

\section{Introduction}

We study $k$-clustering problems in which the objective is minimizing the sum of cluster costs, each of which is a function of its radius, i.e., the maximum distance from the respective cluster center to a point in the cluster. We call this problem class MSRDC (Minimum Sum of Radius-Dependent Costs). Such problems arise, e.g., in wireless network design~\cite{DBLP:journals/cn/Lev-TovP05}. These objectives are different from the extensively studied  $k$-median~\cite{jain1999primal,charikar2002constant,har2004coresets, arya2004local, li2016approximating}  and $k$-means~\cite{har2004coresets,kanungo2004local,vassilvitskii2006k} objectives in that they do not directly count the individual points' distance to their respective cluster center, and it is also different from the well-known $k$-center~\cite{hochbaum1985heuristic,dyer1985heuristic,charikar2001outlier} objective in that not only the maximum cluster cost matters. Much less is known about MSRDC than about the aforementioned problems.

A specific objective from the class of studied objectives is the problem in which $\sum_{i=1}^k r_i^\alpha$ is to be minimized for some $\alpha\in\mathbb{R}_{\geq 1}$~\cite{DBLP:journals/cn/Lev-TovP05}. Especially the version with $\alpha=1$, the min-sum-radii objective, has recieved some attention in the literature. Unfortunately, at least from an exact-computation perspective, the known results are quite negative, even for this restricted problem: The problem is known to be NP-hard, even in metrics of constant doubling dimension and shortest-path metrics of weighted planar graphs~\cite{DBLP:journals/algorithmica/GibsonKKPV10}.

\subsection{Our Contribution}

Following the research directions of Katsikarelis et al.~\cite{DBLP:journals/dam/KatsikarelisLP22} as well as Baker et al.~\cite{DBLP:conf/icml/BakerBHJK020}, who consider (coresets for) clustering in shortest-path metrics of weighted bounded-treewidth graphs, our main result is as follows.

\begin{theorem}\label{thm:main}
    MSRDC in shortest-path metrics parameterized with the treewidth of the underlying graph is in XP.
\end{theorem}

In other words, we show that, for any $\ell\in\mathbb{N}$, there exists a polynomial-time algorithm for MSRDC in shortest-path metrics on graphs with treewidth at most $\ell$.

Here, we consider the arguably most general version of MSRDC in which there is a set of \emph{clients} $C$ that have to be covered with clusters centered at \emph{facilities} $F$, where $C$ and $F$ may intersect arbitrarily.

Our algorithm extends a dynamic program (DP) for tree metrics, which we sketch in the following. For simplicity, at the (unproblematic) cost of possibly turning the metric into a pseudometric  by splitting vertices into two vertices of mutual distance $0$, assume that in the underlying tree $T$ rooted at $r$, each node has at most two children and each node is either a client or a facility.

Consider some node $v\in V(T)\setminus\{r\}$ and some solution $\mathcal{S}$ to the problem. Let $\mathcal{S}'$ be the subsolution of $\mathcal{S}$ that is obtained from restricting $\mathcal{S}$ to the subtree $T'$ of $T$ rooted at $v$. The main observation is the following. Clearly, the number of clusters that $\mathcal{S}'$ uses and its cost is relevant to the remainder of the solution. Beyond that, there are two cases:
\begin{compactitem}
    \item If $\mathcal{S}'$ covers all of $T'$, the only other relevant property of $\mathcal{S}'$ is how large an area $\mathcal{S}'$ covers beyond $v$, i.e., what the largest value of $r_i-d(f_i,v)$ is for a facility $f_i$ in $\mathcal{S}'$ with corresponding radius $r_i$. We call the latter value the \emph{outgoing excess value}.
    \item If $\mathcal{S}'$ does not cover all of $T'$, the only other relevant property of $\mathcal{S}'$ is how large an area the remaining solution must cover within $T'$, again measured from $v$, i.e., what the largest value of $d(v,v')$ is for a $v'\in V(T')$ not covered by $\mathcal{S}$. We call the latter value the \emph{incoming excess value}.
\end{compactitem}
Motivated by this observation, we maintain a DP entry for each combination of $v\in V(T)$, $k'\in\{0,\dots,k\}$, and the different above cases along with all $O(n^2)$  possible incoming/outgoing excess values (all pairwise distances in $T$ suffice) that the solution needs to guarantee.  Note that one cannot simply encode the case in the sign of the excess value (at least not without distinguishing between $-0$ and $+0$) because the case with incoming excess $0$ is different from the case with outgoing excess $0$. What we then save in such an entry is the cheapest solution for the subtree rooted at $v$ using at most $k'$ facilities (if such a solution exists) and respecting the case as well as incoming/outgoing excess value.

The DP can be filled in a bottom-up manner. Leaf nodes, in case they are clients, can be covered at cost $0$ with any incoming excess, and, in case they are facilities, they can have any outgoing excess $c$ at cost $g(c)$ using one facility. For a client with a single child, we only have to translate the incoming/outgoing excess values of the child's DP entries into incoming/outgoing excess values of the new client to obtain the new DP entries. For a facility with a single child, we have to guess whether the facility will be opened and, if so, by how much and then combine with the child's DP entry that will guarantee the desired incoming/outgoing excess (with possibly one facility fewer). The most interesting case is the case of two children, in which we additionally have to guess the (feasible) combination of DP entries that we are using for the children. Note that here it comes in handy that we restricted to at most two children.

In Section~\ref{sec:main}, we extend this approach to bounded-treewidth graphs and formalize it in this general case. We make use of nice tree decompositions\footnote{We refer readers not familiar with the concept to Section~\ref{sec:prelim}.}. We maintain a DP entry for each combination of a bag of the tree decomposition, possible incoming/outgoing excess values for \emph{each} of the vertices in the bag, and the number of facilities that may be used in the corresponding solution. The incoming/outgoing excess values are encoded by a vector of vertices that define the incoming/outgoing excess through their distance to their corresponding vertex in the bag. This is sufficient since only these distances will be relevant as incoming/outgoing excess values. 

In Section~\ref{sec:hardness}, we complement this positive result with the following hardness result, underpinning even more that MSRDC is difficult.

\begin{theorem}\label{hardness}
The min-sum-radii problem is NP-hard even when $|F|=k$.
\end{theorem}

That is, the hardness of the min-sum-radii problem is not merely based on the fact that a set of facilities to be opened has to be selected. We call this problem (with $|F| = k$) MSR-A, which is the assignment version of the min-sum-radii problem. Our proof is inspired by the proof of Gibson et al.~\cite{DBLP:journals/algorithmica/GibsonKKPV10} showing that the (regular) min-sum-radii problem is NP-hard on weighted planar graph metrics and on metrics of bounded doubling dimension.

\subsection{Further Related Work}

For the min-sum-radii objective and $F=C$, there are a $3.504$-approximation algorithm~\cite{DBLP:journals/jcss/CharikarP04} as well as some exact algorithms known: a randomized algorithm with running time $n^{O(\log n \cdot \log \Delta)}$~\cite{DBLP:journals/algorithmica/GibsonKKPV10}, where $\Delta$ is the ratio of the maximum and minimum distance in the given metric, and a polynomial-time algorithm for Euclidean metrics of fixed dimension \cite{DBLP:journals/siamcomp/GibsonKKPV12}, assuming that two candidate solutions can be compared efficiently.

Both the aforementioned results also apply to the corresponding assignment problem, i.e., the version with $|F|=k$ rather than $F=C$. When $g:x\mapsto x^\alpha$ for some $\alpha > 1$, this problem is NP-hard even in the Euclidean plane~\cite{DBLP:conf/esa/BiloCKK05,DBLP:conf/compgeom/AltABEFKLMW06}, and there is a PTAS in Euclidean space of any fixed dimension~\cite{DBLP:journals/cn/Lev-TovP05,DBLP:conf/esa/BiloCKK05}.

Extending our set of objective functions, minimizing $\mathcal{F}(r_1,\dots,r_k)$ has been considered for some arbitrary (polynomial-time computable) symmetric and monotone function $\mathcal{F}:\mathbb{R}^n_+\rightarrow\mathbb{R}_+$, again assuming $F=C$. In the Euclidean plane with $k$ fixed, the problem can be solved exactly in polynomial time~\cite{DBLP:journals/jal/CapoyleasRW91}.

In the min-sum-diameter problem (e.g.~\cite{hansen1987minimum}), there are no cluster centers; the goal is merely to compute a partition $C_1,\dots,C_k$ of $C$ so as to minimize the sum of cluster \emph{diameters}, i.e., maximum inter-client distance within the corresponding cluster. Inapproximability~\cite{DBLP:conf/swat/DoddiMRTW00} and approximation results~\cite{DBLP:conf/swat/DoddiMRTW00,DBLP:conf/swat/BehsazS12} have been obtained.

We discuss two problems for which dynamic programs related to ours have been employed but to obtain different results. In both problems, $F=C$. The first problem is $k$-center with an upper bound on the radius of any cluster (another special case of our setting). When parameterized with clique-width (implying the same result when paramterized with treewidth), the version on graph metrics is in XP and there is an approximation scheme running in FPT time~\cite{DBLP:journals/dam/KatsikarelisLP22}. Furthermore, the min-sum-radii problem with the radii restricted to a given multiset is in XP when parameterized with the number of distinct radii~\cite{DBLP:conf/waoa/LieskovskyS22}. 

The parameter treewidth was independently introduced by several authors, e.g., by Robertson and Seymour~\cite{DBLP:journals/jal/RobertsonS86}, and shown to be NP-complete to compute~\cite{arnborg1987complexity}. Computing a tree decomposition whose width equals the treewidth of the corresponding graph is fixed-parameter tractable when parameterized by this treewidth~\cite{DBLP:journals/siamcomp/Bodlaender96}.

\section{Preliminaries}\label{sec:prelim}
In the MSRDC problem, one is given a metric space $(V=C\cup F,d)$ where $F$ are the facilities and $C$ are the clients, along with $k\in\mathbb{N}$ and oracle access to an increasing function $g:\RR_{\geq 0} \to \RR_{\geq 0}$. 
We let $B(f,r) = \{c\in C: d(c,f)\leq r\}$ for $f\in F$ and a radius $r\in \RR_{\geq 0}$.  
A feasible solution $(S,R)$ consists of up to $k$ facilities $S= (f_1,\ldots,f_k)$ and $k$ radii $R=(r_1,\ldots,r_k)$ such that $C = \bigcup_{i=1}^kB(f_i,r_i)$. Its cost is $\sum_{i=1}^k g(r_i)$.

Since $g$ is increasing, we may restrict to solutions $(S,R)$ such that for each $r$ in $R$ we have $r\in \mathcal{R}:=\{d(u,v) \mid u,v\in V(G)\}$. The reason is that otherwise the radii of the solution could be decreased until they satisfy this condition, without increasing the cost of the solution.  We therefore only consider clusterings in each subinstance where the radii are chosen from $\mathcal{R}$ to reconstruct an optimal solution.

In this paper, we are focussing on the case where $(V,d)$ is the shortest-path metric of an edge-weighted graph $G=(V,E)$ with weight function $w:E\rightarrow \mathbb{R}_{\geq 0}$. 
Note that, strictly speaking, we are faced with a pseudometric rather than a metric because we allow $0$ as an edge weight (which is meaningful, as this way we can model that a point can act as a client and facility).
In the following, we assume $G$ to be a connected graph since otherwise we can use our algorithm on each connected component separately.

A \emph{tree decomposition} of $G$ is a tree $T=(V',E')$ each of whose nodes $t\in V'$ is associated with a \emph{bag} $X_t\subseteq V$ such that the following conditions are fulfilled:
\begin{compactitem}
    \item Each vertex $v\in V$ is contained in at least one bag,
    \item for each edge $e=\{u,v\}\in E$, there exists at least one bag such that both $u$ and $v$ are contained in it, and
    \item for each $v\in V$, the bags containing $v$ induce a connected subtree of $T$.
\end{compactitem}
In this paper, we assume our tree decompositions to be \emph{nice} (e.g.,~\cite{DBLP:books/sp/CyganFKLMPPS15}). For that to be the case, $T$ has to be a binary tree with root $r$ such that $X_r=\emptyset$. Furthermore, for each $t\in V'$: 
\begin{compactitem}
    \item $t$ is a leaf and $X_t=\emptyset$ ($t$ is a \emph{leaf node}).
    \item $t$ has only a single child $t'\in V'$, and $X_{t} = X_{t'} \cup \{v\}$ for some $v\in V\setminus X_{t}$ ($t$ is an \emph{introduce node}).
    \item $t$ has only a single child $t'\in V'$, and $X_t = X_{t'}\setminus  \{v\}$ for some $v\in X_{t'}$ ($t$ is a \emph{forget node}).
    \item $t$ has two children $t_1,t_2\in V'$ and $X_t = X_{t_1} = X_{t_2}$ ($t$ is a \emph{join node}).
\end{compactitem}

The \emph{width} of such a tree decomposition is $\max_{t\in V(T)}\vert X_t \vert -1$. Note that, for every tree decomposition, there exists a nice tree decomposition of the same graph with the same width, and it can be computed in polynomial time~\cite{DBLP:books/sp/CyganFKLMPPS15}. The \emph{treewidth} of $G$ is the minimal width of any (nice) tree decomposition of $G$.

For each node $t\in V'$ we define $V_t$ to contain all vertices in $X_t$ or in $X_{t'}$ for some descendant of $t$.
Likewise we define $C_t=C\cap V_t$, $F_t=F\cap V_t$ and $G_t=(V_t,E_t)$, where $E_t$ only contains edges from $E(G)$, where both endpoints are in $V_t$. Finally, for brevity we write $[k]:= \{1,\ldots,k\}$ for $k\in\mathbb{N}$ and denote the null tuple as $()$.

\section{A polynomial-time algorithm for bounded-treewidth graphs}
\label{sec:main}

    In this section, we show Theorem~\ref{thm:main}. Our algorithm is based on a dynamic program (DP). In the first subsection, we describe what the DP entries are supposed to contain, and in the second subsection we describe how to compute them and prove the correctness.

\subsection{Entries of the Dynamic Program}

There is a DP entry for each combination of the following values:
\begin{compactitem}   
    \item A bag  $t\in V(T)$ with $X_t = \{v_1,\dots,v_{|X_t|}\}$,
    \item a vector of vertices $c \in (F \cup C)^{X_t}$.
    \item a vector of directions $\circ  \in\{\uparrow, \downarrow\}^{ X_t }$, and
    \item a number of facilities $k'\in\{0,\dots,k\}$.
\end{compactitem}

Informally, what we store in an entry $D[t,c,\circ, k']$ is the cheapest solution for $G_t$ using at most $k'$ facilities and obeying the \emph{excess requirements} given by $\circ$ and the distances $d(v,c_v)$ for $v \in X_t$. For a given tuple $(t,c,\circ,k')$ we  require a specific set of vertices to be covered, using only facilities from $F_t$:
$$C^{c, \circ,\text{rem}}_t:=C_t\setminus\bigcup_{v:\circ_v=\downarrow}B(v,d(v,c_v)).$$
Therefore we say that a solution $(S,R)$ is \emph{covering for} $(t,c,\circ,k')$ if $C^{c, \circ,\text{rem}}_t \subseteq \bigcup_{f\in S}B(f,r_f)$.
Let $\SSS = (S,R)$ with $S=(f_1,...,f_{\ell})\in F_t^{\ell}$ for $\ell\leq k'$ and radii $R\in \mathcal{R}^{S}$ be a solution for $(t,c,\circ, k')$. 
Then the \emph{excess coverage} of $v \in X_t$  w.r.t. $\SSS$ and $(t,c,\circ, k')$ is 
\begin{align*}
&e_v^{t,c,\circ}({\SSS}):= \\
&\max \big\{\max_{f\in S} ( r_f - d(f, v) ), \max_{\substack{w \in X_t\\ 
\circ_w = \da
}} ( d(v,c_w) - d(w, v) )\big\}.
\end{align*}

We call a covering solution \emph{feasible for} $(t,c,\circ,k')$ if for each  $v \in X_t$ the excess requirements are satisfied, i.e., $e_{v}^{t,c,\circ}({\SSS}) \geq d(v,c_v)$ for every $v\in X_t$ (Note that vertices with incoming excess satisfy this condition by themselves). We may call a solution simply feasible (and omit $(t,c,\circ,k')$) if it is clear from the context.

The \emph{border vertex} of $v \in X_t$ w.r.t.\ a covering solution $\SSS$ for $(t,c,\circ, k')$ is the furthest-away vertex $q\in C \cup F$ from $v$ within distance $e_v^{t,c,\circ}(\SSS)$, i.e., 
\begin{align*}
&b_v^{t,c,\circ}({\SSS})\in
 \arg\max_{\substack{q \in C\cup F\\ 
d(v,q) \leq e_v^{t,c,\circ}({\SSS})
}} d(v,q).
\end{align*}
At last, for every $v\in V_t$ we define $F_v^{t,c, \circ}(\mathcal{S}) := \{f \in S \mid r_f-d(f,v)=e_v^{t,c, \circ}(\SSS)\}$ and $C_v^{t,c, \circ}(\mathcal{S}) := \{q\in X_t \mid \circ_v=\da \wedge d(q,c_q)-d(q,v)=e_v^{t,c, \circ}(\SSS)\}$ as the set of vertices which provide the excess coverage in vertex $v \in V_t$ by some radius or incoming excess, respectively. 
Notice that for every covering solution $\mathcal{S}$ and $v\in C_t \cup X_t$ we have $\vert F_v^{t,c, \circ}(\mathcal{S}) \cup C_v^{t,c, \circ}(\mathcal{S}) \vert \geq 1$.  

Now the entry $D[t,c,\circ, k']$ contains  a feasible solution $(S,R) \in F_t^\ell \times \mathcal{R}^\ell$ for $\ell \leq k'$, minimizing $\sum_{i = 1}^\ell g(r_i)$, if such a solution exists. As a shorthand, we refer to the cost of  this solution by $D^{\text{val}}[t, c, \circ, k']$.
If no feasible solution exists, we have $D[t,c,\circ, k'] = \text{NIL}$  and $D^{\text{val}}[t, c, \circ, k']=\infty$.
Note that, with this definition, the entry $D[r,(),(),k]$  contains an optimal solution for the original instance as $r$ is the root node of the tree decomposition. We will use $\mathcal{S}^\emptyset := ((),())$ as abbreviation for the empty solution.

Given two tuples $(t_1,c',\circ',k_1)$ and $(t_2,c'',\circ'',k_2)$  with solutions $(S_1,R_1) = D[t_1,c',\circ',k_1]$  and  $(S_2,R_2) = D[t_2,c'',\circ'',k_2]$ respectively. If none of $(S_1,R_1)$ and $(S_2,R_2)$ are NIL, their \emph{combined solution} is obtained by concatenating the lists of facilities and the lists of radii, otherwise it is NIL.

\subsection{Computing the Entries}\label{sec:entries_lemmas}
In the following, we show how to compute the correct entries from the bottom to the top of the tree $T$. We distinguish the four different node types.

\subsubsection{Leaf nodes} Recall that, for each leaf node $t$, we have $V_t= \emptyset$. Therefore $\Cp = \emptyset$, and the optimal solution has cost $0$ for any $k'\geq 0$. Therefore it is correct to set $D[t, (), (), k'] = \mathcal{S}^\emptyset$. 

\subsubsection{Introduce nodes} 

Consider any entry $D[t,c,\circ,k']$ such that $t$ is an introduce node with child~$t'$ and $X_{t} = X_{t'} \cup \{v\}$. We consider every solution $D[t', c', \circ', k']$ for $ c'\in V^{ X_{t'}  }$, $\circ'\in \{\ua, \da\}^{ X_{t'}}$ and check if it is  feasible for $(t,c,\circ, k')$. 
If $k'\geq 1$ and $v\in F$ we also consider solutions where $v$ is chosen as a facility with radius $r_v\in \mathcal{R}_v:=\{d(v,q)\mid q\in V\}$, combine it with $D[t', c', \circ', k'-1]$ and check if it is feasible. Among all such feasible solutions, we select one with minimum cost.

\begin{restatable}{lem}{IntroduceNodesLemma}\label{lem:introduce_node}
Given a tuple $(t,c,\circ,k')$ where $t$ is an introduce node and assume all entries of its descendant are correct. Then $D[t, c, \circ, k']$ is correct.
\end{restatable}

\begin{proof}
Since we only consider feasible solutions, the optimal solution is at least as good as the one retrieved by our algorithm.

Next we need to show that the solution computed by our algorithm is at least as good as the optimal solution. If no feasible solution exists, then $D^\text{val}[t, c, \circ, k'] = \infty$ is clearly correct. So assume that there exists at least one solution which is feasible.
Let $\SSS^* = (S^*, R^*)$ with $S^*=(f_1^*,...,f_\ell^*) \in F_t^\ell$ and $R^*\in \mathcal{R}^{S^*}$ for $\ell\leq k'$ be an optimal solution for $(t,c,\circ,k')$.

First consider the case $v \notin S^*$, so we have $S^* \subseteq F_{t'}$. 
For every vertex $w\in X_{t'}$ we set $c_{w}' = b_{w}^{t,c,\circ}(\mathcal{S}^*)$; if $\vert C_w^{t,c, \circ}(\mathcal{S}^*) \vert = 0$ (implying $\vert F_w^{t,c, \circ}(\mathcal{S}^*) \vert \geq 1$), we set  $\circ_{w}'=\ua$, and otherwise we set $\circ_{w}'=\da$.
 Let $\mathcal{S} = (S, R) = D[t', c', \circ', k']$ be the solution computed by our algorithm. Note that $\mathcal{S}$ is not NIL because the solution $\mathcal{S}^*$ is also feasible for $(t', c', \circ', k')$, so at least one feasible solution exists.
We first argue that every vertex $w\in \Cp  \setminus \{v\}$
is also contained in $\Cp[t'][c',\circ']$ and therefore covered in the instance corresponding to $(t,c,\circ,k')$ by $\mathcal{S}$: 
Let $q\in C_{t'} \setminus \Cp[t'][c',\circ']$ be arbitrary. 
This vertex is already covered by some vertex $w\in X_{t'}$ with
$d(w,c'_w)\geq d(w,q)$.
We  show that $q$ is also already covered in $(t,c,\circ,k')$ as follows: 
Since $w$ has some incoming excess
$d(w,c'_w)$ in $(t',c',\circ',k')$, we know that there exists some vertex
$z\in X_{t}$ providing this excess ($w$ could be $z$ itself).
Therefore we have $d(z,c_z) \geq d(z,w) + d(w,c'_w)  \geq d(z,w) + d(w,q) \geq d(z,q)$, so the vertex $q$ must also be covered in $(t,c,\circ,k')$. 
This means that $\Cp \setminus \{v\}\subseteq \Cp[t'][c',\circ']$. 

Additionally, to be accepted as a feasible solution, we need $e_w^{t,c,\circ}(\mathcal{S})\geq d(w,c_w)$
for every vertex $w\in X_t$.
For every vertex $w\in X_{t'}$, this follows from the entries of $c'$ being set to the border vertices 
of the optimal solution.
If $\circ_v = \da$, we also know that $v$ receives the excess, so at last, assume $\circ_v = \ua$.
Since $v$ is not in $S^*$,
there either exists some facility $f^*$ in $S^*$ with $f^* \in F_v^{t,c,\circ}$ or vertex $v^* \in C_v^{t,c,\circ}$.
If $\vert C_v^{t,c,\circ} \vert \geq 1$ any solution (and therefore also the solution computed by $D[t',c',\circ', k']$) will fulfill the excess requirements in $v$, since the excess coverage in $v$ can only increase but never decrease.
So in the following assume that $\vert C_v^{t,c,\circ} \vert = 0$ and the excess coverage is contributed only by a facility $f^*$. This means $e_v^{t,c,\circ}(\mathcal{S}^*) = r_{f^*} - d(f^*,q) - d(q,v)$
for some $q\in X_{t'}$, since  any path connecting $f^*$ with $v$ has to pass through at least one vertex in $X_{t'}$ by the definition of a tree decomposition.
If $\circ_q'=\da$, we can argue analogously to the previous case above that there exists some $z\in X_{t'}$ with incoming excess which provides the excess coverage to $v$.
If $\circ_q'=\ua$, there exists some facility $f'\in S$ with $r_{f'}\geq d(f',q) + d(q,b_q^{t,c,\circ}(\mathcal{S}^*))$ since at least one facility is required to cover $q$ with at least $d(q,b_q^{t,c,\circ}(\mathcal{S}^*)) \geq d(q,v) + d(v,c_v)$.
So we know that the excess coverage $e_q^{t,c,\circ}(\mathcal{S})$ also fulfills the excess requirements of $v$. 
We can thus conclude that $D^\text{val}[t,c,\circ, k']$ is at most $D^\text{val}[t',c',\circ', k']$, which is equal to the cost of $\mathcal{S}^*$.

Next, assume $v \in S^*$. Let $r_v\in \mathcal{R}_v$ be its radius and assume we guess this value in our algorithm. 
Let $\SSS^*_{v}$
be the solution $\SSS^*$ after removing facility $v$. 
For a vertex $p \in X_t$, we compare the excess coverage $e_p^{t,c,\circ}(\SSS^*)$ to $e_p^{t, c, \circ}(\SSS^*_{v})$ and construct $c'$ and $\circ'$ for $t'$ depending on both the old excess coverage and the new one.
Since removing a facility from the solution can only decrease the excess coverage in any vertex $p\in X_t$, we have $e_p^{t, c, \circ}(\SSS^*_{v}) \leq e_p^{t, c, \circ}(\SSS^*)$. 
For every $p \in X_t$ we set $c_{p}'= b_{p}^{t,c,\circ}(\SSS^*)$. If $e_{p}^{t, c, \circ}(\SSS^*) = e_{p}^{t,c,\circ}(\SSS^*_{v})$ and if $\vert C_{p}^{t,c,\circ}(\mathcal{S}^*) \vert = 0$ (implying $\vert F_{p}^{t,c, \circ}(\mathcal{S}^*) \vert \geq 1$), we set $\circ_{p}' = \ua$, otherwise  $\circ_{p}' =\da$. Notice that $\SSS^*_{v}$ is a feasible solution for our constructed instance $D[t', c', \circ', k'-1]$ since, for every vertex $p\in X_{t'}$ which is not covered by $v$ (and therefore the real excess did not change), we have $c_p' = b_{p}^{t,c,\circ}(\SSS^*) $.
Again let $\SSS$ = $D[t', c', \circ', k'-1]$. By previous arguments we analogously get $\Cp\setminus B(v,r_v) \subseteq \Cp[t'][c',\circ']$. Since we guessed $r_v$, we can conclude that $\SSS$ combined with $v$ covers all vertices in $\Cp$. 

Regarding the excess requirements we distinguish two cases. First consider vertices $p \in X_t$ with $v \in F_{p}^{t,c, \circ}(\mathcal{S}^*)$, here the excess requirements are satisfied trivially by including $v$ with $r_v$. For the vertices $p \in X_t$ with $v \notin F_{p}^{t,c, \circ}(\mathcal{S}^*)$ the excess requirements are satisfied by analogous argumentation as in the previous case with $v \notin S^*$. This shows the feasibility of  $\SSS$ combined with $(v,r_v)$. 
Thus, we conclude that $D^\text{val}[t,c,\circ, k']$ is at most $D^\text{val}[t', c', \circ', k'-1] + g(r_v) \leq D^\text{val}[t, c, \circ, k']$, which is equal to the cost of $\SSS^*$.
\end{proof}

\subsubsection{Forget nodes}

We define $c' \in V^{X_{t'}}$  s.t.\ $c'_w=c_w$ for $w\in X_{t}$ and $c'_v = v$ and $\circ',\circ'' \in \{\da,\ua\}^{X_{t'}}$ s.t. $\circ'_w = \circ''_w=\circ_w$ for $w \in X_t$ and $\circ''_v = \da$, $\circ'_v = \ua$. Assume we have some node $t$ with $X_t = X_{t'}\setminus \{v\}$ for some $v\in X_{t'}$.
Since $v$ cannot be adjacent to any other vertex from $V\setminus V_t$, we only need to ensure that it is appropriately used in the subsolution from $D[t, c, \circ, k']$:
\begin{equation*}
    D[t, c, \circ, k']
    = \begin{cases}
    D[t', c', \circ', k'] & v\in \Cp, \\
    D[t', c', \circ'', k'] & \text{otherwise}.
    \end{cases}
\end{equation*}

\begin{restatable}{lem}{forgetNodesLemma}
Let $(t, c, \circ , k')$ be a tuple of parameters where $t$ is a forget node and assume all entries for its descendant are correct. Then $D[t, c, \circ, k']$ is correct.
\end{restatable}

\begin{proof}
First assume we have $v\in \Cp$.
Notice that both instances corresponding to the tuples $(t,c,\circ, k')$ and $(t, c', \circ', k')$ are equivalent since the sets of feasible solutions are identical.
Therefore we have $D^\text{val}[t, c, \circ, k'] = D^\text{val}[t', c', \circ', k']$.

Now assume $v \notin \Cp$.  In this case we either have $v\in C$ and $v$ is already covered by some incoming excess, or we have $v \in F$. 
Therefore both instances corresponding to the tuples $(t,c,\circ, k')$ and $(t, c', \circ'', k')$ are equivalent because we have $d(v,w)>0$ for all $w\in V\setminus \{v\}$. In this case we also have $D^\text{val}[t,c, \circ, k'] = D^\text{val}[t', c',\circ'', k']$.
\end{proof}

\subsubsection{Join nodes} 

Let us consider a join node where for some node $t$ we have $X_t=X_{t_1}=X_{t_2}$ for two children nodes $t_1,t_2$.
We work in a similar fashion by iterating over all possible inputs $c'\in (F\cup C)^{X_t},\circ'\in\{\ua,\da\}^{X_t},k_1\leq k'$ and $c''\in (F\cup C)^{X_t}, \circ''\in\{\ua,\da\}^{X_t},k_2\leq k'$ with $k_1+k_2 = k'$ and select the cheapest feasible combined solution. 

\begin{restatable}{lem}{joinLemma}
Let $(t, c, \circ , k')$  be a tuple of parameters where $t$ is a join node and assume all entries for its descendants are correct. Then $D[t, c, \circ, k']$ is correct.
\end{restatable}

\begin{proof}
Since we only consider feasible solutions, the optimal solution is clearly at least as good as the one retrieved by our algorithm. 

We need to show that the solution computed by our algorithm is at least as good as the optimal solution. 
If no feasible solution exists, then $D^\text{val}[t, c, \circ, k'] = \infty$ is clearly correct.
So assume that there exists a feasible solution.
Let $\SSS^* = (S^*, R^*)$ with $S^*=(f_1^*,...,f_\ell^*) \in F_t^\ell$ and $R^*\in \mathcal{R}^{S^*}$ for $\ell\leq k'$ be an optimal solution for $(t,c,\circ,k')$. For the two subinstances defined by $t_1$ and $t_2$, let $\SSS_{t_1}^*=(S_1^*,R_1^*)$ and $\SSS_{t_2}^*=(S_2^*,R_2^*)$ be the optimal facilities in  $V_{t_1}$ and $V_{t_2}\setminus X_t$, respectively, with their corresponding radii.
Additionally let $k_1=\vert S_1^* \vert$ and $k_2=\vert S_2^* \vert$.  In the following, we construct $c', \circ', k_1$ and $c'', \circ'', k_2$ for $t_1$ and $t_2$ such that $\SSS_{t_1}^*$ and $\SSS_{t_2}^*$ are feasible for the respective tuples $(t_1, c', \circ', k_1)$ and $(t_2, c'', \circ'', k_2)$. 
Then we argue that the solution obtained by combining  $D[t_1, c', \circ', k_1]$ and $D[t_2, c''\circ'', k_2]$ is feasible for $(t, c, \circ, k')$. 

For every vertex $w\in X_t$ we set $c_{w}'=c_{w}''=b_{w}^{t,c,\circ}(\SSS^*)$.
Additionally, if $\vert C_{w}^{t,c,\circ}(\SSS^*)\vert \geq 1$, we set $\circ_{w}'=\circ_{w}''=\da$. If this is not the case we set
$\circ'_w = \ua$ iff $\vert F_{w}^{t,c,\circ}(\SSS^*) \cap V_{t_1}\vert \geq 1$ and $\circ''_w = \ua$ iff $\vert F_{w}^{t,c,\circ}(\SSS^*) \cap (V_{t_2} \setminus X_t)\vert \geq 1$. Since every vertex has to be covered, we know that, if the first condition did not apply, either $\circ_w'=\ua$ or $\circ_w''=\ua$.
Per construction of $c',\circ',c'',\circ''$, we can see that $\SSS_1^*$ is feasible for entry $D[t_1, c', \circ', k_1]$ and $\SSS_2^*$ is feasible  for entry $D[t_2, c'', \circ'', k_2]$.

Let $\SSS$ be the solution obtained by combining $\SSS_1 = D[t_1,c',\circ',k_1]$ and $\SSS_2 = D[t_2,c'',\circ'', k_2]$. Since we have $k_1 + k_2 = k'$, we know that $\SSS$ contains at most $k'$ facilities.  

First we show that every vertex in $\Cp$ is covered by $\SSS$.  Assume for contradiction there exists some vertex $p\in \Cp$ which is not covered by the solution $\SSS$. 
The feasibility of $\SSS_1$ and $\SSS_2$ imply that $p\notin  \Cp[t_1][c',\circ'] \cup \Cp[t_2][c'',\circ'']$. 
Assume $p\in V_{t_1}$ (the case $p \in V_{t_2}\setminus X_t$ is analogous). 
Since $p \notin \Cp[t_1][c',\circ']$ there exists some vertex $q \in X_{t_1}$ which covers $p$, i.e., $\circ_q'=\da$ and $d(q,c_q') \geq d(q,p)$.
By construction of $c',\circ'$ we know that $\circ_q'=\da$ iff
\begin{itemize}
    \item[(i)] $\vert C_q^{t,c,\circ}(\SSS^*)\vert \geq 1$, or
    \item[(ii)] the three conditions $\vert C_{q}^{t,c,\circ}(\mathcal{S}^*)\vert = 0$, $|F_q^{t,c,\circ}(\mathcal{S}^*) \cap V_{t_1}\vert = 0$ and $\vert F_q^{t,c,\circ}(\mathcal{S}^*) \cap (V_{t_2}\setminus X_t)\vert \geq 1$ hold.
\end{itemize}

If (i), let $w\in C_q^{t,c,\circ}(\SSS^*)$, then $d(w,c_w) \geq d(w, q) + d(q,c_q') \geq d(w, q) + d(q, p)\geq d(w, p).$
Since $\circ_w = \da$ we get that $p \notin \Cp$, which is a contradiction. 

If (ii), we know that $\circ_q'' = \ua$ and $c_q'' = b_q^{t,c,\circ}(\SSS^*)$. 
If the excess requirements are not already satisfied by incoming excess, i.e., $e_q^{t_2,c'',\circ''}(\SSS^\emptyset) < d(q,c_q'')$, we are done because then there has to exist some facility $f\in S_2$ s.t.\ $r_f \geq d(f, q) + d(q, c_q'') \geq d(f,q) + d(q, p) \geq d(f, p)$, so $p$ is also covered by $f$ in the constructed solution $\SSS$.
So assume $e_q^{t_2,c'',\circ''}(\SSS^\emptyset) \geq d(q,c_q'')$, which means that there exists some vertex $y\in X_{t_2}$ with $\circ_y''=\da$ and $d(y,c_y'') \geq d(y, q) + d(q, c_q'')$.
Per construction, $\circ_y''=\da$ means that we are in one of two cases: Either $\vert C_y^{t,c,\circ}(\SSS^*)\vert \geq 1$, or $\vert C_{y}^{t,c,\circ}(\mathcal{S}^*)\vert = 0$, $\vert F_y^{t,c,\circ}(\mathcal{S}^*) \cap (V_{t_2}\setminus X_t) \vert = 0$ and $\vert F_y^{t,c,\circ}(\mathcal{S}^*) \cap V_{t_1} \vert \geq 1$ hold.
The first case cannot happen as this implies $\vert C_q^{t,c,\circ}(\SSS^*)\vert \geq 1$ because then there exists some vertex $w\in X_t$ with $\circ_w=\da$ and $d(w, c_w) \geq d(w, y) + d(y, c_y'') \geq d(w, y) + d(y, q) + d(q, c_q'')\geq d(w, q) + d(q, c_q'')$, so the incoming excess in $w$ is large enough to satisfy the excess requirement of $q$.
Therefore assume the second case, which includes condition $\vert F_y^{t,c,\circ}(\mathcal{S}^*) \cap V_{t_1} \vert \geq 1$, so there exists some facility $f_1^*\in S_1^*$ with $r_{f_1^*}\geq d(f_1^*, y) + d(y, c_y'')$.
But this implies that $\vert F_q^{t,c,\circ}(\SSS^*)\cap V_{t_1}\vert \geq 1$ because $r_{f_1^*}\geq d(f_1^*, y) + d(y, c_y'') \geq d(f_1^*, y) + d(y, q) + d(q, c_q'') \geq d(f_1^*, q) + d(q, c_q'')$, so there exists some facility in $V_{t_1}$ which provides the excess coverage of $q$.
This is a contradiction to our choice of $\circ_q'=\da$.
Therefore, $p$ is covered by $\SSS$.

At last, we need to guarantee that the combined solution $\mathcal{S}$ fulfills the excess requirements, i.e., $e_{p}^{t,c,\circ}(\mathcal{S}) \geq d(p,c_p)$ for each $p\in X_t$. 
For each vertex $p$ with $\circ_p=\da$ this is trivial.
Consider the case $\circ_p=\ua$  
and assume that the excess requirement is not fulfilled by some incoming excess, otherwise any solution fulfills the excess requirement of $p$. 
Since $\SSS^*$ is feasible there exists a facility $f^* \in S^*$ with radius $r_{f^*} \geq d(p,c_p) + d(f,p)$.  
By our construction we either set $\circ_p'=\ua$ or $\circ_p''=\ua$. Additionally we set $c'_p=c_p''=b_{p}^{t,c,\circ}(\SSS^*)$ for which it holds that $d(p,b_{p}^{t,c,\circ}(\SSS^*))\geq d(p,c_p)$. 
Assume in the following that $\circ_p'=\ua$, the other case follows analogously.
If $e_p^{t_1,c',\circ'}(\SSS^\emptyset) < d(p, c_p')$ we know that there has to exist some facility $f\in S_1$ with $r_f \geq d(f,q) + d(q, c'_p)\geq d(f,q) + d(q, c_q)$, so we fulfill the excess requirements.
So assume in the following $e_p^{t_1, c', \circ'}(\SSS^\emptyset) \geq d(p, c_p')$.
Like in the previous part of the proof we can show that in this case we will have some facility in $S_2$, which fulfills the excess requirements of $p$.

Finally we show that $\SSS$ is at least as good as the optimal solution. Let $\text{OPT} = \text{OPT}_1 + \text{OPT}_2$ be the cost of $\SSS^*$, with $\text{OPT}_1=\sum_{r_i\in R_1^*}g(r_i)$ as the cost of $\SSS_1^*$ and $\text{OPT}_2=\sum_{r_i \in R_2^*}g(r_i)$ as the cost of $\SSS^*_2$. 
Since $\mathcal{S}_1^*$ and $\mathcal{S}_2^*$ is feasible for $D[t_1,c', \circ', k_1]$ and $D[t_2,c'', \circ'', k_2]$, respectively, it follows from the optimality of the DP entries that  $D^\text{val}[t_1, c', \circ', k_1] \leq \text{OPT}_1$ and $D^\text{val}[t_2, c'', \circ'', k_2] \leq \text{OPT}_2$. 
Since the cost of the combined solution of two DP entries is at most as large as the sum of costs of both solutions, we get $\text{OPT}=\text{OPT}_1+\text{OPT}_2 \geq  D^\text{val}[t_1, c', \circ', k_1] + D^\text{val}[t_2, c'', \circ'', k_2] \geq D^\text{val}[t,c, \circ, k']$.
\end{proof}

\subsection{Time complexity}\label{sec:time_complexity}
We can compute the metric closure in $O(\vert V\vert^3)$ and for a given instance $(G,k)$ and solution $\mathcal{S}=(S,R)$ check in time $O(\vert V\vert \cdot k)$ if the solution is feasible.
For a bag $t$ we have DP entries for every $i \in \{0,\dots,k\}$, $c\in V^{X_t}$ and $\circ \in \{\da,\ua\}^{X_t}$. Therefore, there are in total $(2\cdot |V|)^{|X_t|} \cdot k \in O((2\cdot |V|)^{\ell} \cdot k)$ many entries for a bag $t$.

Next we analyze the running time for the individual node types.
For an introduce node, we need to consider all ways to open the introduced vertex $v$ as a facility with a radius defined by a distance to a vertex in $\mathcal{R}_v$ and check if the solution is feasible. Since $|\mathcal{R}_v| \in O(|V|)$, we get for introduce nodes a running time $O((2\cdot |V|)^{\ell} \cdot (k+1) \cdot |V| \cdot (|V| \cdot k)) \in O(2^{\ell} \cdot |V|^{(\ell+2)}\cdot k^2)$.
For a forget node we get a running time of $O(\vert V\vert\cdot k)$ for checking if the solution is feasible.
For a join node, we consider every possible combination of subinstances,
leading to $O((2 \cdot |V|)^{2\cdot \ell}\cdot k)$ entries to check for feasibility, which yields a running time
of $O(2^{2\ell} \cdot |V|^{2\cdot \ell+1}\cdot k^2 )$.

Therefore we can upper bound the time to fill all entries corresponding to some node $t$ by interpreting $t$ as a join node, i.e., $O((2^{2\ell} \cdot |V|^{2\ell+1}\cdot k^2) \cdot ( (2\cdot |V|)^{\ell} \cdot k) ) ) = O(2^{3\ell} \cdot |V|^{3\ell +1} \cdot k^3)$.
Since there are $O(\ell \cdot \vert V\vert)$ nodes in a tree decomposition~\cite{DBLP:books/sp/CyganFKLMPPS15}, we get an overall running time of $O(\ell\cdot \vert V \vert \cdot (2^{3\ell} \cdot |V|^{3\ell +1} \cdot k^3)) = O(\ell \cdot 2^{3\ell} \cdot |V|^{3 \ell + 2} \cdot k^3)$.

\section{Hardness of MSR-A}\label{sec:hardness}
In this section we prove Theorem~\ref{hardness}. The idea is similar to the proof of Gibson et al.~\cite{DBLP:journals/algorithmica/GibsonKKPV10} who show that a problem similar to MSR-A, the \kk{}-cover problem, is NP-hard by a reduction from a special variant of the \textit{planar 3-SAT problem}. We, however, first reduce from the regular 3-SAT problem.

\begin{proof}
Let $\Phi = (X,C)$ be an instance of $3$-SAT where $X=\{x_1,...,x_{n}\}$ are the variables and $\mathcal{C}=\{C_1,...,C_m\}$ are the clauses, consisting of exactly three literals each. Consider the \emph{incidence graph} $G' = (V', E')$ that has a vertex $\ell_i$ for each literal, a vertex $c_j$ for each clause and has an edge between $\ell_i$ and $c_j$ iff literal $\ell_i$ is contained in clause $C_j$. Let $G = (V,E)$ arise from $G'$ by adding a vertex $y_i$ and edges $\{x_i,y_i\}$ and $\{\overline{x_i},y_i\}$ for all $i\in \{1,\dots, n\}$.
Furthermore, define a weight function $w:E\to \mathbb{R}_{\geq0}$ by setting $w(e)=2^{i-1}$ for all edges incident to $x_i$ or $\overline{x_i}$. Let $L$ be the set of vertices corresponding to literals. 

Our instance $\mathcal{I}_\Phi$ for the MSR-A problem is then given by $L$ as the facilities, $V$ as the clients and a distance metric $d$ obtained by the metric closure of $(G,w)$.  See Figure~\ref{fig:hardness_example} for an illustration.
We show that the 3-SAT instance $\Phi$ is satisfiable iff $\mathcal{I}_\Phi$ has a solution of cost at most $2^n-1$.

Assume we have some truth assignment satisfying $\Phi$. For all $i\in\{1,\dots n\}$, either $x_i$ or $\overline{x_i}$ must be true. Open the facility corresponding to the true literal with radius $2^{i-1}$ and set the other radius to $0$. Since we have exactly one true literal for each $i$, all vertices $y_i$ are covered by some facility. Since the assignment is satisfying, we also know that all clauses must contain at least one literal, so each clause vertex is also covered by some facility. The cost of this solution is $\sum_{i=1}^{n}2^{i-1} = 2^n - 1$.

Now assume that we have a solution to the MSR-A instance with cost at most  $2^n -1$.
Note that we may assume that, for every facility $f\in L$, its radius is either $0$ or equal to some distance to a $y_i$ or clause vertex.  Let $L'\subseteq L$ denote the set of facilities with positive radius. We obtain a truth assignment for $\Phi$ by setting all literals corresponding to vertices in $L'$ to true and all other literals to false. 
Now for each variable $x_i$, exactly one of the literals $x_i$ and $\overline{x_i}$ must be true. For the sake of a contradiction, let us first assume that there exists some $i$ for which neither $x_i$ nor $\overline{x_i}$ is true. Consider the largest $i$ for which this happens. This means that both $\ell_i$ and $\overline{\ell_i}$ have radius $0$ and therefore $y_i$ has to be covered by some other facility with index  $i'$. By construction of the graph, the radius of that facility then has to be at least $2^{i'} + 2\cdot 2^{i-1}.$ But this would lead to an overall cost of at least $2 \cdot 2^{i-1} + \sum^{n-1}_{j = i} 2^j = 2^n$ which contradicts the upper bound on the cost of the solution.
On the other hand, assume that there exists some $i$ for which both $x_i$ and $\overline{x_i}$ are true. Again take the largest such $i$. Then the cost of the solution is again lower-bounded by $2 \cdot 2^{i-1} + \sum^{n-1}_{j = i} 2^j = 2^n$.

Thus it follows that exactly one literal is true for each variable, which makes the assignment feasible. Furthermore, the solution for the MSR-A instance is feasible and therefore covers all clause vertices. This translates to every clause containing at least one true literal, making the assignment satisfying as well.
\end{proof}
\begin{figure}
\centering
\includegraphics[width=.6\linewidth]{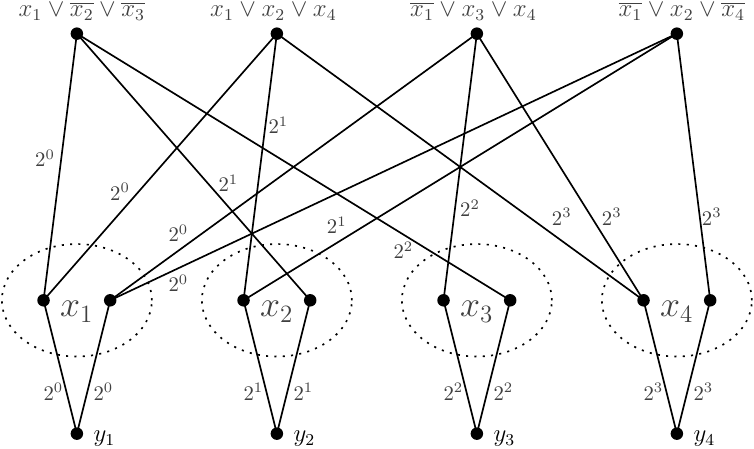}
\caption{Example transformation of a 3-SAT instance with 4 variables and 4 clauses to a MSR-A instance.
}
\label{fig:hardness_example}
\end{figure}
Using the same technique as in \cite{DBLP:journals/algorithmica/GibsonKKPV10}, we can transform any planar 3-SAT instance into an equivalent planar instance for our MSR-A problem. The reason is that the graph we construct is a subgraph of the graph constructed in~\cite{DBLP:journals/algorithmica/GibsonKKPV10}.

\begin{corollary}
The MSR-A problem is NP-hard even for planar metrics.
\end{corollary}

\section{Conclusion}

We have presented a polynomial-time exact algorithm for clustering to minimize the sum of radius-dependent costs in graphs of bounded treewidth. We note that our techniques can be generalized to handling cost functions that also depend on the facility. While the number of dynamic-programming entries that our approach relies on is inherently exponential in the treewidth of the underlying graph, the arguably most interesting follow-up question is whether the problem is fixed-parameter tractable w.r.t.\ treewidth.

\bibliographystyle{elsarticle-num} 
\bibliography{biblio}

\end{document}